\documentclass[12pt]{article}%
\usepackage{amsfonts}
\usepackage{amsmath}
\usepackage{amssymb}
\usepackage{graphicx}%
\providecommand{\U}[1]{\protect\rule{.1in}{.1in}}
\newtheorem{theorem}{Theorem}
\newtheorem{acknowledgement}[theorem]{Acknowledgement}

\newtheorem{definition}[theorem]{Definition}

\newtheorem{lemma}[theorem]{Lemma}

\newtheorem{proposition}[theorem]{Proposition}
\newtheorem{remark}[theorem]{Remark}

\newenvironment{proof}[1][Proof]{\noindent\textbf{#1.} }{\ \rule{0.5em}{0.5em}}
\begin{document}

\title{Generalized Marginals of tie -Wigner Distribution nand Lagrangian Tomography}
\author{Maurice A. de Gosson\\University of Vienna\\Faculty of Mathematics (NuHAG)\\1090 Vienna, AUSTRIA}
\maketitle

\begin{abstract}
We develop a symplectic formulation of the Lagrangian Radon transform based on
the transitive action of the symplectic group on the Lagrangian Grassmannian.
This geometric framework naturally extends the classical notion of the
marginals of the Wigner transform, viewed as elementary tomograms, to
arbitrary Lagrangian subspaces. The resulting generalized tomography provides
a unified treatment of position and momentum measurements related by
symplectic rotations. As applications, we revisit the Pauli reconstruction
problem and establish reconstruction formulas for density operators from
generalized Lagrangian tomograms.

\end{abstract}

\textbf{Keywords}: Wigner function; Radon transform; quantum tomography;
Lagrangian plane; density operator

\textbf{MSC classification 2020}: 52A20, 52A05, 81S10, 42B35

\section{Introduction}

In \cite{Michor} Marmo et \textit{al. }introduced a functional transformation
which they call \textquotedblleft Lagrangian Radon transform\textquotedblright%
; they construct an operator$\mathit{,}$ which sends a function on
$\mathbb{R}^{2rn}$ to its integrals over all affine Lagrangian subspaces in
$\mathbb{R}^{2rn}$. In this context also see the work \cite{manko,mankobis} by
Man'ko and his collaborators. Their constructions use heavy machinery from
symplectic geometry, in particular properties of the affine symplectic group.
whose action on the affine Lagrangian Grassmannian is studied. This allows
these authors to prove deep and difficult results involving at some point the
Siegel half-plane. The role of the unitary subgroup appears to have received
comparatively little attention, but this fact drastically simplifies the study
of the Lagrangian Radon transform. In the present paper we redefine the
Lagrangian Radon transform in terms of the identification of the unitary group
with a maximally compact subgroup of the symplectic group, and recover its
properties; in particular a simple inversion formula is proposed. This
approach was already proposed in our paper \cite{entropy}, and the present
work substantially extends these results.

\section{Background}

The idea of using partial data\ to reconstruct a function goes back to the
1917 work \cite{Radon} by J. Radon. His original aim was to reconstruct a
function of two variables $x,p$ from its integrals on arbitrary straight lines
$ax+bp=X$ in the $x,p$ plane. In two papers Ibort \textit{et al}.
\cite{Ibort,Ibortbis} considerably refine this approach and set out to study
what they call the \textquotedblleft tomographic picture of quantum
mechanics\textquotedblright\ by using a version of the Radon transform they
define by the somewhat heuristic formula%
\begin{equation}
R_{\rho}(X,a,b)=\int_{\mathbb{R}^{2}}\rho(x,p)\delta(X-ax-bp)dpdx
\label{Radon1}%
\end{equation}
where $\rho$ is a quasi-probability \ and $a,b$ arbitrary real numbers. Of
course such a definition is \textit{stricto sensu }usable only under certain
smoothness conditions on $\rho$. While the papers \cite{Ibort,Ibortbis} make
heavy use of representation theory and are accessible to a physically-minded
readership, we will in this paper study the Radon transform on $\mathbb{R}%
^{n}$ from the point of view of harmonic analysis by making systematically use
the theory of the metaplectic group together with the elementary theory of
Lagrangian subspaces of the standard symplectic space. This new approach
allows us to give rigorous proofs and give new insights on the theory of the
Radon transform. Here is the main idea in a nutshell. In \cite{Pauli} W. Pauli
asked whether one can determine a function $\psi\in L^{2}(\mathbb{R}^{n})$
only knowing the squared moduli $|\psi(x)|^{2}$ and $|\widehat{\psi}(p)|^{2}$
($\widehat{\psi}$ the Fourier transform of $\psi$). The answer is
\textquotedblleft no\textquotedblright\ as can be shown by producing
elementary counterexamples. The next step is to remark that if we assume in
addition $\psi,\widehat{\psi}\in L^{1}(\mathbb{R}^{n})$ then $|\psi(x)|^{2}$
and $|\widehat{\psi}(p)|^{2}$ are the marginal distributions of the Wigner
transform of $\psi$:%
\begin{equation}
\int_{\mathbb{R}^{n}}W\psi(x,p)dp=|\psi(x)|^{2}\text{ },\text{ }%
\int_{\mathbb{R}^{n}}W\psi(x,p)dx=|\widehat{\psi}(p)|^{2}. \label{wigmarg}%
\end{equation}

The quantum covariance matrix is an essential object in quantum mechanics and
quantum information science, used for analyzing and characterizing various
properties of quantum systems. For instance, it allows the characterization of
Gaussian states (which are uniquely determined by their covariance matrix), to
prove entanglement, \ to detect noise and dissipation in quantum systems. It
also provides, in quantum experiments, the covariance matrix provides a
practical method to measure and analyze global observables.

In our study of the reconstruction of Gaussian states \cite{gopol1} we
introduced a variant of the Radon transform \cite{Radon}. This transform,
which we call Lagrangian Radon transform was defined as follows:

Recall that the Wigner transform of a function $\psi\in L^{2}(\mathbb{R}^{n})$
is defined by
\begin{equation}
W\psi(x,p)=\left(  \frac{1}{2\pi\hbar}\right)  ^{n}\int_{\mathbb{R}^{n}%
}e^{-\frac{i}{\hbar}p\cdot y}\psi(x+\tfrac{1}{2}y)\overline{\psi(x-\tfrac
{1}{2}y)}dy \label{Wigner}%
\end{equation}
and if $\psi,\widehat{\psi}\in L^{1}(\mathbb{R}^{n})\cap L^{2}(\mathbb{R}%
^{n})$ it satisfies the marginal properties (\ref{wigmarg}). These formulas
can be interpreted as the integrals of the Wigner transform $W\psi$ along the
coordinates planes $\ell_{X}=\mathbb{R}^{n}\times0$ and $\ell_{P}%
=0\times\mathbb{R}^{n}$, respectively. The idea underlying the Radon transform
is now to calculate the integrals of $W\psi$ along \textit{all} Lagrangian
subspaces $\ell$ of the symplectic space $\mathbb{R}^{n}\times\mathbb{R}^{n}$
(i.e. subspaces obtained by symplectic rotations from $\ell_{X}$ (or $\ell
_{P}$); the knowledge of all these integrals then allow the determination of
$\psi$. Now, the main observation is that these integrals can be defined by
\[
R(U,\psi)(x)=\int_{\mathbb{R}^{n}}W\psi(U^{-1}(x,p))dp
\]
where $U$ is a symplectic rotation taking $\ell_{X}$ to $\ell$; using the
symplectic covariance property$W\psi\circ U^{-1}=W(\widehat{U}\psi)$ where
$\widehat{U}$ is a metaplectic operator covering $U$ we get%
\[
R(U,\psi)(x)=\int_{\mathbb{R}^{n}}W(\widehat{U}\psi)(x,p))dp=|\widehat{U}%
\psi(x)|^{2}%
\]
which is what we call the \textit{Lagrangian Radon transform} of $\psi$. So,
viewed from the metaplectic perspective, the Radon transform is nothing more
than the datum of all the marginals of the Wigner transform of $\widehat{U}%
\psi$ when $\widehat{U}$ describes the group of all symplectic rotations; that
group, denoted by $U(n)$ is isomorphic to the unitary group $U(n,\mathbb{C})$.
(See Corbett's \cite{Corbett} review of Pauli's problem.)

\paragraph{Notation and terminology}

We will denote by $\sigma$ the standard symplectic form on $\mathbb{R}%
^{2n}\equiv\mathbb{R}_{x}^{n}\times\mathbb{R}_{p}^{n}$ that is, in matrix
form\textit{,} $\sigma(z,z^{\prime})=Jz\cdot z^{\prime}$ where $z=(x,p)^{T}$,
$z^{\prime}=(x^{\prime},p^{\prime})^{T}$ and $J=%
\begin{pmatrix}
0 & I\\
-I & 0
\end{pmatrix}
$. The scalar product of two vectors $u,v\in\mathbb{R}^{m}$ is written $u\cdot
v$. We denote by $\operatorname*{Sym}_{++}(m,\mathbb{R})$ the convex set of
all symmetric positive definite real $m\times m$ matrices.

The group of all automorphisms of the symplectic space $(\mathbb{R}%
^{2n},\sigma)$ leaving the standard symplectic form invariant is called the
(standard) symplectic group and is denoted by $\operatorname*{Sp}(n)$. The
additive group of all real symmetric $m\times m$ matrices is denoted by
$\operatorname*{Sym}(m,\mathbb{R})$.

A subspace of $\mathbb{R}^{2n}$ with dimension $n$ on which $\sigma$ vanishes
identically is called a Lagrangian subspace (or plane). The set of all
Lagrangian subspaces of $(\mathbb{R}^{2n},\sigma)$ will be denoted by
$\operatorname*{Lag}(n)$, and is called the Lagrangian Grassmannian of the
symplectic space $(\mathbb{R}^{2n},\sigma).$

\section{Symplectic Geometry and Lagrangian Frames}

For a detailed study of the topics of this section see for instance
\cite{Folland,Birk,Leray}. Also see the monographs \cite{GS1,GS2} by
Guilllemin n Sternberg

\subsection{ The groups $\operatorname*{Sp}(n)$ and $U(n)$}

The symplectic group $\operatorname*{Sp}(n)$ consists of all linear
\ automorphisms $S$ of the symplectic space $(\mathbb{R}^{2n},\sigma$ which
preserve the symplectic form $\sigma$ that is
\[
\sigma(Sz,Sz^{\prime})=\sigma(z,z^{\prime})
\]
for all $z,z^{\prime}\in\mathbb{R}^{2n}$. The group $\operatorname*{Sp}(n)$ is
a connected classical Lie groups and $\operatorname*{Sp}(1)$ consists of all
real matrices with determinant one. Identifying $S\in\operatorname*{Sp}(n)$
wit its matrix in the canonical basis of $\mathbb{R}^{2n}$ we have
$S\in\operatorname*{Sp}(n)$ if and only if $SJS^{T}=S^{T}JS=J$; \ it follows
that $\operatorname*{Sp}(n)$ is a closed subgroup of $GL(2n,\mathbb{R})$ and
hence a classical Lie group. Writing $S=%
\begin{pmatrix}
A & B\\
C & D
\end{pmatrix}
$, where the \textquotedblleft blocks\textquotedblright\ $A,B,C,D$ being
$n\times n$ matrices, we have $S\in\operatorname*{Sp}(n)$ if and only if
\cite{GS2,Birk}
\begin{align}
A^{T}C,\text{ }B^{T}D\text{ \ \textit{are symmetric, and} }A^{T}D-C^{T}B &
=I\label{cond12}\\
AB^{T},\text{ }CD^{T}\text{ \ \textit{are\ symmetric, and} }AD^{T}-BC^{T} &
=I.\label{cond22}%
\end{align}
One shows that the group $\operatorname*{Sp}(n)$ is generated by the standard
symplectic matrix $J$ together with the matrices%
\[
V_{P}=%
\begin{pmatrix}
I_{n\times n} & 0\\
-P & I_{n\times n}%
\end{pmatrix}
\text{ \ },\text{ \ }M_{L}=%
\begin{pmatrix}
L^{-1} & 0\\
0 & L^{T}%
\end{pmatrix}
\]
where $P\in\operatorname*{Sym}(n,\mathbb{R})$ and $L\in GL(n,\mathbb{R})$.

A subgroup of $\operatorname*{Sp}(n)$ of particular interest is the image
$U(n)$ in $\operatorname*{Sp}(n)$ of the unitary group $U(n,C)$ by the
monomorphism
\[
\iota:u=+iB\longmapsto U=%
\begin{pmatrix}
A & B\\
-B & A
\end{pmatrix}
\]
(the relations (\ref{cond12}) (\ref{cond22})) are satisfied since $u^{\ast
}u=uu^{\ast}u=I$). The elements of $U)n)$ are symplectic rotations:%
\begin{equation}
U(n)=\operatorname*{Sp}(n)\cap O(2m,\mathbb{R}). \label{USPO}%
\end{equation}
It follows from the conditions that $U\in U(n)$ if and only if they satisfy
the equivalent conditions
\begin{align}
A^{T}B\text{ \textit{symmetric and }}A^{T}A+B^{T}B  &  =I\label{u1}\\
AB^{T}\text{ \textit{symmetric and }}AA^{T}+BB^{T}  &  =I. \label{u2}%
\end{align}

\subsection{The metaplectic representation of $\operatorname*{Sp}(n)$}

The symplectic group $\operatorname*{Sp}(n)$ is connected and contractible to
its maximal subgroup
\[
U(n)=\operatorname*{Sp}(n)\cap O(2n(
\]
the latter is  isomorphic to the unitary group $U(n,C)$ hence the group
isomorphisms
\[
\pi_{1}(\operatorname*{Sp}(n))\simeq\pi_{1}(U(n,C))\simeq(\mathbb{Z},+).
\]
It follows that $\operatorname*{Sp}(n)$ has covering groups
$\operatorname*{Sp}_{q}(n)$ of all orders $q=2,3,...,+\infty$. It turns out
that the double cover $\operatorname*{Sp}_{2}(n)$ has a unitary representation
in $L^{2}(\mathbb{R}^{n})$ by the metaplectic group $\operatorname*{Mp}(n)$.
The covering mapping
\begin{equation}
\pi_{\operatorname*{Mp}}:\operatorname*{Mp}(n)\longrightarrow
\operatorname*{Sp}(n)\text{ \ , \ }\pi_{\operatorname*{Mp}}(\widehat{S}%
)=S\label{pimp}%
\end{equation}
satisfies $\operatorname*{Ker}(\pi_{\operatorname*{Mp}})=\{-I,I\}$ and is
adjusted so that
\[
\pi_{\operatorname*{Mp}}(\widehat{J})=J\text{ \ if \ }\widehat{J}%
\psi(x)=\left(  \tfrac{1}{2\pi i\hbar}\right)  ^{n/2}\int_{\mathbb{R}^{n}%
}e^{-\frac{i}{\hbar}x\cdot x^{\prime}}\psi(x^{\prime})dx^{\prime}%
\]
and one shows that $\operatorname*{Mp}(n)$ is generated by $\widehat{J}$
together with the unitary automorphisms
\[
\widehat{V}_{P}\psi(x)=e^{-\frac{i}{2}Px\cdot x}\psi(x)\text{ \ ,
\ }\widehat{M}_{L,m}\psi(x)=i^{m}\sqrt{|\det L|}\psi(Lx)
\]
where $P\in\operatorname*{Sym}(n,\mathbb{R})$ and $L\in GL(n,\mathbb{R})$; the
integer $m$ corresponds to a choice of $\arg\det L$. The operators
$\widehat{V}$ and $\widehat{M}$ cover the symplectic automorphisms $V_{P}$ and
$M_{L,m}$ defined above:
\[
\pi_{\operatorname*{Mp}}(\widehat{V}_{P})=V_{P}\text{ \ \textit{and} \ }%
\pi_{\operatorname*{Mp}}(\widehat{M}_{L,m})=M_{L}.
\]

Assume that $S=%
\begin{pmatrix}
A & B\\
C & D
\end{pmatrix}
$ is a free symplectic matrix, \textit{i.e.}. that $\det B\neq0$. Then the
metaplectic operators $\pm\widehat{S}$ covering $S$ are given by
\begin{equation}
\widehat{S}_{\mathcal{A},m}\psi(x)=\left(  \tfrac{1}{2\pi\hbar}\right)
^{n/2}i^{m-n/2}\sqrt{|\det B^{-1}|}\int_{\mathbb{R}^{n}}e^{\frac{i}{\hbar
}\mathcal{A}(x,x^{\prime})}\psi(x^{\prime})dx^{\prime}\label{qft1}%
\end{equation}
where
\begin{equation}
\mathcal{A}(x,x^{\prime})=\frac{1}{2}DB^{-1}x\cdot x-B^{-1}x\cdot x^{\prime
}+\frac{1}{2}B^{-1}Ax^{\prime}\cdot x^{\prime}\label{wfree}%
\end{equation}
and the integer $m$ is $0$ modulo $2$ if $\det B>0$ and modulo $2$ if $\det
B>0$. Notice that $\mathcal{A}$ generates $S$ in the sense that
$(x,p)=SA(x^{\prime},p^{\prime})$ if and only if $p=\partial_{x}%
\mathcal{A(}x,x^{\prime})$ and $p^{\prime}=-\partial_{x^{\prime}}%
\mathcal{A(}x,x^{\prime})$. When $S=U\in U(n)$ the metaplectic operators
$\pm\widehat{U}$ are generalized fractional Fourier transformers; for instance
for $n=1$ and $U=%
\begin{pmatrix}
\cos\varphi & \sin\varphi\\
-\sin\varphi & \cos\varphi
\end{pmatrix}
$ we have, for $\sin\varphi>0$,
\begin{gather*}
\widehat{U}\psi(x)=\left(  \tfrac{1}{2\pi i\hbar}\right)  ^{1/2}\sqrt{\frac
{1}{\sin\varphi}}\int_{\mathbb{-\infty}}^{\infty}e^{\frac{i}{\hbar}%
\mathcal{A}(x,x^{\prime})}\psi(x^{\prime})dx^{\prime}\\
\mathcal{A}(x,x^{\prime})=\frac{1}{2}\cot\varphi(x^{2}+x^{\prime2})-\frac
{1}{\sin\varphi}xx^{\prime}-\frac{1}{2}\cot\varphi(x^{\prime}\cdot x^{\prime
}).
\end{gather*}

\subsection{Lagrangian planes}

A linear subspace of $(\mathbb{R}^{2n},\sigma)$ with dimension $n$ is on which
the symplectic form $\sigma$ vanishes identically is called a Lagrangian
plane. A typical example is provided by the subspace with coordinates
$x_{1},...,x_{k},p_{k+1},...,p_{n}$ with $1\leq k<n$. The set of all
Lagrangian subspaces of $(\mathbb{R}^{2n},\sigma)$ will be denoted by
$\operatorname*{Lag}(n)$, and is called the Lagrangian Grassmannian of the
symplectic space $(\mathbb{R}^{2n},\sigma)$. There is a natural transitive
action
\[
\operatorname*{Sp}(n)\times\operatorname*{Lag}(n)\longrightarrow
\operatorname*{Lag}(n).
\]
We will denote by $\ell_{X}$ and $\ell_{P}$ the position and momentum
Lagrangian planes $\mathbb{R}^{n}\times0$ and $0\times0\mathbb{R}^{n}$,
respectively. The unitary group $U(n)$ (and hence symplectic group
$\operatorname*{Sp}(n)$) acts transitively on $\operatorname*{Lag}(n)$.

The transitivity of this action allows to identify \cite{Birk}
$\operatorname*{Lag}(n)$ with closed subset $W(n)$ of all symmetric elements
of $U(n)$; for this one shows that the \textquotedblleft Souriau
mapping\textquotedblright\ \cite{Souriau,Birk}:%
\[
\ell=\iota(u)\ell_{P}\in\operatorname*{Lag}(n)\longmapsto\iota(u)\iota
(u^{T})\in W(n)
\]
is a bijection $\operatorname*{Lag}(n)\longrightarrow W(n)$. Using this
bijection one equips $\operatorname*{Lag}(n)$ with a natural topology.

Notice that the transitivity of the action%
\begin{equation}
U(n)\times\operatorname*{Lag}(n)\longrightarrow\operatorname*{Lag}(n)
\label{ulag}%
\end{equation}
implies that every $\ell\in\operatorname*{Lag}(n)$ has an equation%
\begin{equation}
Ax+Bp=0\text{ \ },\text{ \ }(x,p)\in\mathbb{\ell} \label{axbp}%
\end{equation}
with $A,B$ as in (\ref{u1}), (\ref{u2}) above.

We will call \textit{Lagrangian frame} any pair $(\ell,\ell^{\prime})$ of
transverse Lagrangian subspaces, that is $(\ell,\ell^{\prime})\in
\operatorname*{Lag}(n)\times\operatorname*{Lag}(n)$ and $\ell\cap\ell^{\prime
}=0$ (equivalently $\ell\oplus\ell^{\prime}=\mathbb{R}^{2n}$). The pair
$(\ell_{X},\ell_{P})$ will be called the \textit{canonical frame}. \ We will
denote by $\mathcal{F}_{\operatorname*{Lag}}(n)$ the set of all Lagrangian
frames of $(\mathbb{R}^{2n},\sigma)$. An essential result is:

\begin{lemma}
\label{LemmaTrans}The symplectic group action
\begin{gather*}
\operatorname*{Sp}(n)\times\mathcal{F}_{\operatorname*{Lag}}(n)\longrightarrow
\mathcal{F}_{\operatorname*{Lag}}(n)\\
S((\ell,\ell^{\prime})\longmapsto(S\ell,S\ell^{\prime})
\end{gather*}
is transitive. In particular, every Lagrangian frame $(\ell,\ell^{\prime}%
)\in\mathcal{F}_{\operatorname*{Lag}}(n)$ can be obtained from the canonical
frame $(\ell_{X},\ell_{P})$ by a symplectic automorphism.
\end{lemma}

\begin{proof}
Let $(\ell_{1},\ell_{1}^{\prime})$ and $(\ell_{2},\ell_{2}^{\prime})$ be two
Lagrangian frames. Choose a basis $(e_{1i})_{1\leq1\leq n}$ of $\ell_{1}$ and
a basis $(f_{1j})_{1\leq j\leq n}$ of $\ell_{1}^{\prime}$ whose union
$(e_{1i})_{1\leq1\leq n}\cup(f_{1j})_{1\leq j\leq n}$ is a symplectic basis,
that is $\sigma(e_{1i},e_{1j})=\sigma(f_{1i},f_{1j})=0$ and $\sigma
(f_{1i},e_{1j})=\delta_{ij}$ for $1\leq i,j\leq n$. Similarly, choose bases
$(e_{2i})_{1\leq1\leq n}$ and $(f_{2j})_{1\leq j\leq n}$ of $\ell_{2}$ and
$\ell_{2}^{\prime}$ whose union is also a symplectic basis. The linear
automorphism of $\mathbb{R}^{2n}$ defined by $S(e_{1i})=e_{2i}$ and
$S(f_{1i})=f_{2i}$ for $1\leq i\leq n$ is in $\operatorname*{Sp}(n)$ and we
have $(\ell_{2},\ell_{2}^{\prime})=(S\ell_{1},S\ell_{1}^{\prime})$.
\end{proof}

\section{The Lagrangian Radon Transform}

\subsection{Lagrangian Radon transform: first definitions}

For $\ell\in\operatorname*{Lag}(n)$ and $z=(x,p)\in\mathbb{R}^{2n}$ we set
\[
\ell(z)=\ell+z.
\]
Assuming that $\ell$ is given by the equation
\[
\ell:Ax+Bp=0
\]
with $A^{T}B$ symmetric and $A^{T}A+B^{T}B=I$ we parametrize $\ell(z)$ by
\begin{equation}
x(u)=-B^{T}u+x\text{ \ },\text{ \ }p(u)=A^{T}u+p \label{parameter1}%
\end{equation}
and hence
\[
\ell=U^{-1}\ell_{P}=U^{T}\ell_{P}\text{ \ },\text{ \ }U=%
\begin{pmatrix}
A & B\\
-B & A
\end{pmatrix}
\in U(n).
\]

\begin{definition}
The Lagrangian Radon transform (for short:\ LRT) $R_{\ell}W\psi$ is defined by
integral
\begin{equation}
R_{\ell}W\psi(z)=\int_{\mathbb{R}^{n}}W\psi(-B^{T}u+x,A^{T}u+p)du.
\label{integral}%
\end{equation}
We will write%
\[
R_{\ell}W\psi(z)=\int_{\ell(z)}W\psi(\lambda)d\lambda.
\]

\end{definition}

Applying this definition to the case where $\psi\in L^{1}(\mathbb{R}^{n})\cap
L^{2}(\mathbb{R}^{n})$ and choosing $\ell=\ell_{P}$ we have $B=0$ and
$A=I_{n\times n}$ so that, in view of the marginal properties (\ref{wigmarg}),%
\[
R_{\ell_{P}}W\psi(z)=\int_{\mathbb{R}^{n}}W\psi(x,u+p)du=|\psi(x)|^{2}.
\]
Similarly, choosing $\ell=\ell_{P}$, and $A=0$, $B=I_{n\times n}$,%
\begin{equation}
R_{\ell_{X}}W\psi(z)=\int_{\mathbb{R}^{n}}W\psi(-u+x,p)du=|\widehat{\psi
}(p)|^{2}. \label{rlx}%
\end{equation}

Let us prove that :

\begin{proposition}
For fixed $z=(x,p)$ the integral (\ref{integral}) is independent of the choice
of parametrization (\ref{parameter1}).
\end{proposition}

\begin{proof}
We begin by noting that the Lagrangian subspace $\ell:Ax+Bp=0$ is the image of
$\ell_{X}=\{(u,0):u\in\mathbb{R}^{n}\}$ by the symplectic rotation%
\[
V^{T}=%
\begin{pmatrix}
-B^{T} & -A^{T}\\
A^{T} & -B^{T}%
\end{pmatrix}
\in U(n).
\]
Let a new parametrization of $\ell(z)$ be
\begin{equation}
x^{\prime}(u)=-B^{\prime T}u+x\ ,\ \ p^{\prime}(u)=A^{\prime T}u+p \label{re}%
\end{equation}
the matrices $A^{\prime}$ and $B^{\prime}$ satisfying relations similar to
those of $A$ and $B$, and let $V^{\prime}$ be the corresponding symplectic
rotation. The product $V^{-1}V^{\prime}$ leaves $\ell_{X}$ invariant; since it
is a symplectic rotation we have $^{-1}V^{\prime}=%
\begin{pmatrix}
H & 0\\
0 & H
\end{pmatrix}
$ with $H\in O(n,\mathbb{R})$ hence the reparametrization (\ref{re}) is given
by%
\[
x^{\prime}(u)=-B^{T}Hu+x\ ,\ p^{\prime}(u)=A^{T}Hu+p
\]
leading to the same value of the integral \ref{integral}) since $d(Hu)=du$.
\end{proof}

The following result relates the integral (\ref{integral}) to the notion of
marginal value of the Wigner transform:

\begin{proposition}
\label{Proprawig}Let $\ell\in\operatorname*{Lag}(n):Ax+Bp$, that is
$\ell=U^{-1}\ell_{P}$ were%
\[
U=%
\begin{pmatrix}
A & B\\
-B & A
\end{pmatrix}
\in\operatorname*{Sp}(n)\cap O(2n,\mathbb{R)}.
\]
For $\psi\in L^{1}(\mathbb{R}^{n})\cap L^{2}(\mathbb{R}^{n})$, we have
\begin{equation}
R_{\ell}W\psi(x,p)=|\widehat{U}\psi(Ax+Bp)|^{2} \label{integralbis}%
\end{equation}
where $\widehat{U}\in\operatorname*{Mp}(n)$ covers the symplectic rotation $U$.
\end{proposition}

\begin{proof}
Using the symplectic covariance relation
\begin{equation}
W\psi(U^{^{-1}}z)=W(\widehat{U}\psi)(z) \label{fa}%
\end{equation}
valid for any metaplectic operator \cite{Birk,WIGNER}, we have
\begin{align*}
R_{\ell}W\psi(z)  &  =\int_{\mathbb{R}^{n}}W\psi(U^{-1}\left[
(0,u)+U(x,p))\right]  du\\
&  =\int_{\mathbb{R}^{n}}W(\widehat{U}\psi)(0,u)+U(x,p))du\\
&  =\int_{\mathbb{R}^{n}}W(\widehat{U}\psi)(Ax+Bp,Bx-Ap+u)du\\
&  =\int_{\mathbb{R}^{n}}W(\widehat{U}\psi)(Ax+Bp,u)du
\end{align*}
hence (\ref{integralbis}) in view of the first marginal property
(\ref{wigmarg}).
\end{proof}

The Lagrangian Radon transform can be extended by replacing the symplectic
rotation $U$ with an arbitrary invertible matrix $%
\begin{pmatrix}
A & B\\
-B & A
\end{pmatrix}
$ with $A^{T}B=B^{T}A$. This extension will allow us to obtain a simple
inversion formula. Defining
\[
U_{\Lambda}=%
\begin{pmatrix}
A & B\\
-B & A
\end{pmatrix}
\text{ \ },\text{ \ }\Lambda=(A^{T}A+B^{T}B)^{1/2}%
\]
and assuming that $\det\Lambda\neq0$, we have the factorization
\[
U_{\Lambda}=%
\begin{pmatrix}
\Lambda & 0\\
0 & \Lambda
\end{pmatrix}%
\begin{pmatrix}
\Lambda^{-1}A & \Lambda^{-1}B\\
-\Lambda^{-1}B & \Lambda^{-1}A
\end{pmatrix}
\]
where%
\begin{equation}
U=%
\begin{pmatrix}
\Lambda^{-1}A & \Lambda^{-1}B\\
-\Lambda^{-1}B & \Lambda^{-1}A
\end{pmatrix}
\in\operatorname*{Sp}(n)\cap O(2n,\mathbb{R})
\end{equation}
has the inverse
\begin{equation}
U^{-1}=%
\begin{pmatrix}
A^{T}\Lambda^{-1} & -B^{T}\Lambda^{-1}\\
B^{T}\Lambda^{-1} & A^{T}\Lambda^{-1}%
\end{pmatrix}
.
\end{equation}

\begin{definition}
The matrices $U_{\Lambda}$ and $U$ being defined as above, the extended
Lagrangian Radon transform of $W\psi$ is defined by%
\begin{equation}
R_{\ell,\Lambda}(W\psi)(x,p)=\det\Lambda^{-1}|\widehat{U}\psi(\Lambda
^{-1}(Ax+Bp)|^{2} \label{xtend}%
\end{equation}
where $\widehat{U}\in\operatorname*{Mp}(n)$ covers the symplectic rotation $U$.
\end{definition}

When $U_{\Lambda}=U$ then $\Lambda=I$ and (\ref{xtend}) reduces to formula
(\ref{integralbis}). \ /

\subsection{The Radon inversion formula}

The main result of this section is that the Wigner transform can be
reconstructed from its Radon transform:

\begin{proposition}
\label{Propinverse}Viewing the matrices $A$ and $B$ as vectors of
$\mathbb{R}^{n^{2}}$ we have
\begin{equation}
W\psi(x,p)=\left(  \frac{1}{2\pi\hbar}\right)  ^{2n^{2}}\int_{\mathbb{R}%
^{n(2n+1)}}R_{\ell,\Lambda}(W\psi)(y,p)e^{\frac{i}{\hbar}(y-Ax-Bp)}dydAdB
\label{inverse}%
\end{equation}
where $dA=%
%TCIMACRO{\tprod _{i,j=1}^{n}}%
%BeginExpansion
{\textstyle\prod_{i,j=1}^{n}}
%EndExpansion
da_{ij}$ and $dB=%
%TCIMACRO{\tprod _{i,j}^{n}}%
%BeginExpansion
{\textstyle\prod_{i,j}^{n}}
%EndExpansion
db_{ij}$ if $A=(a_{ij})_{1\leq i,j\leq n}$, $B=(b_{ij})_{1\leq i,j\leq n}$.
\end{proposition}

\begin{proof}
Let us denote by $W$ the integral in the right-hand side of (\ref{inverse}).
We have, using definition of the Lagrangian Radon transform) and thereafter
marginal condition (\ref{wigmarg}),
\begin{align*}
W &  =\det\Lambda^{-1}\int|\widehat{U}\psi(\Lambda^{-1}y)|^{2}e^{\frac
{i}{\hbar}(y-Ax-Bp)}dydAdB\\
&  =\det\Lambda^{-1}\int W(\widehat{U}\psi)(\Lambda^{-1}y,q)e^{\frac{i}{\hbar
}(Xy-Ax-Bp)}dqdydAdB\\
&  =\int W\psi(U^{-1}(y,q))e^{\frac{i}{\hbar}(\Lambda y-Ax-Bp)}dqdydAdB\\
&  =\int W(U^{-1}(y,q))e^{\frac{i}{\hbar}(\Lambda y-Ax-Bp)}dqdydAdB,
\end{align*}
that is, by inverting $U$:%
\begin{multline*}
W=\int W\psi(A^{T}\Lambda^{-1}y-B^{T}\Lambda^{-1}q,B^{T}\Lambda^{-1}%
y+A^{T}\Lambda^{-1}q)\\
\times e^{i(\Lambda y-Ax-Bp)}dqdydAdB.
\end{multline*}
Setting $Y=A^{T}\Lambda^{-1}y-B^{T}\Lambda^{-1}q$ and $Z=B^{T}\Lambda
^{-1}y+A^{T}\Lambda^{-1}q)$ we have $dYdZ=dqdy$ and hence
\[
W=\int W\psi(Y,Z)e^{i(A(Y-x)+B(Z-p))}dYdZdAdB.
\]
Integration of the exponential with respect to the variables $A$ and $B$
yields%
\[
\int_{\mathbb{R}^{2n^{2}}}e^{i(A(Y-x)+B(Z-p))}dAdB=(2\pi\hbar)^{2n^{2}}%
\delta(Y-x,Z-p)
\]
and hence
\begin{align*}
W &  =(2\pi\hbar)^{2n^{2}}\int W\psi(Y,Z)\delta(Y-x,Z-p)|dYdZ\\
&  =(2\pi\hbar)^{2n^{2}}W\psi(x,p)
\end{align*}
from which the inversion formula (\ref{inverse}) immediately follows.
\end{proof}

\section{ Reconstruction of Gaussian Functions}

\subsection{The functions $\psi_{X,Y}$ and their Wigner transforms}

The most general functions whose Wigner transforms are positive are the
Gaussians $\psi_{X,Y}$ defined by%
\begin{equation}
\psi_{X,Y}(x)=\left(  \tfrac{1}{\pi\hbar}\right)  ^{n/4}(\det X)^{1/4}%
e^{-\tfrac{1}{2\hbar}(X+iY)x\cdot x} \label{psixy}%
\end{equation}
and their translations $x\longmapsto\psi_{X,Y}(x-x_{0})$, where $X,Y\in
\operatorname*{Sym}(n,\mathbb{R})$ are $X>0$. This function is normalized to
unity: $||\psi_{X,Y}^{\gamma}||_{L^{2}}=1$ and its Wigner transform is given
by \cite{Birk,WIGNER}
\begin{equation}
W\psi_{X,Y}(z)=\left(  \tfrac{1}{\pi\hbar}\right)  ^{n}e^{-\tfrac{1}{\hbar
}Gz\cdot z} \label{wxy}%
\end{equation}
where $G$ is the positive definite and symplectic matrix
\begin{equation}
G=%
\begin{pmatrix}
X+YX^{-1}Y & YX^{-1}\\
X^{-1}Y & X^{-1}%
\end{pmatrix}
.
\end{equation}
Observe that $G$ factorizes as $G=S^{T}S$ \ where \
\begin{equation}
S=%
\begin{pmatrix}
X^{1/2} & 0\\
X^{-1/2}Y & X^{-1/2}%
\end{pmatrix}
\in\operatorname*{Sp}(n). \label{sts}%
\end{equation}
It follows that%
\begin{equation}
W\psi_{X,Y}(S^{-1}z)=\left(  \tfrac{1}{\pi\hbar}\right)  ^{n}e^{-\tfrac
{1}{\hbar}z\cdot z}=W\phi_{0}(z) \label{ouips}%
\end{equation}
where $\phi$ is the standard Gaussian%
\[
\phi_{0}(x)=\psi_{I,0}(x)=(\pi\hbar)^{-n/4}e^{-|x|^{2}/2\hbar}.
\]
Let $\widehat{S}\in\operatorname*{Mp}(n)$ be any of the two metaplectic
operators covering $S;$ recalling the symplectic covariance property
\begin{equation}
W(\widehat{S}\psi=W\psi(S^{-1}z)\text{, \ }\psi\in L^{2}(\mathbb{R}^{n^{2}})
\label{wigcov}%
\end{equation}
of the Wigner function, we have:

\begin{lemma}
Every Gaussian $\psi_{X,Y}$ is the image (up to a constant unimodular factor)
the image of the standard Gaussian $\phi$ by a metaplectic operator. In fact
\begin{equation}
\psi_{X,Y}=i^{\gamma}\widehat{V}_{Y}\widehat{M}_{X^{1/2},0}\phi_{0}.
\label{psifi}%
\end{equation}

\end{lemma}

\begin{proof}
Combining (\ref{ouips}) and (\ref{wigcov})we have $W(\widehat{S}\psi
_{X,Y})=W\phi_{0}$ hence $c\widehat{S}\psi_{X,Y}=\phi_{0}$ for some complex
number $c$ such that $|c|=1$. Formula (\ref{psifi}) is obvious.
\end{proof}

It easily follows from this result that the metaplectic group
$\operatorname*{Mp}(n)$ acts transitively on the set of Gaussians $\psi_{X,Y}$.

\subsection{The reconstruction problem for $\psi_{X,Y}$; Pauli Partners}

The Pauli problem has a solution for Gaussians $\widehat{S}\psi_{X,Y}$. One
can prove this using the following method (see \cite{gopol1}). Setting
$G=\frac{\hbar}{2}\Sigma^{-1}$ we can rewrite (\ref{wxy}) as
\begin{equation}
W\psi_{X,Y}(z)=\left(  \tfrac{1}{\pi\hbar}\right)  ^{n}e^{-\tfrac{1}{2}%
\Sigma^{-1}z\cdot z}%
\end{equation}
which identifies $\Sigma$ with the covariance matrix of the probability
distribution $W\psi_{X,Y}$. Let us set%
\begin{equation}
\Sigma=%
\begin{pmatrix}
\Sigma_{XX} & \Sigma_{XP}\\
\Sigma_{PX} & \Sigma_{PP}%
\end{pmatrix}
; \label{cov1}%
\end{equation}
we have
\[
\Sigma_{XX}=\Sigma_{XX}^{T},\text{ \ }\Sigma_{PP}=\Sigma_{PP}^{T},\text{
\ }\Sigma_{XP}=\Sigma_{PX}^{T}\text{ }%
\]
where we are using the notation
\[
\Sigma_{XX}=(\sigma_{x_{j}x_{k}})_{1\leq j,k\leq n},\text{ }\Sigma
_{PP}=(\sigma_{p_{j}p_{k}})_{1\leq j,k\leq n},\text{ }\Sigma_{PP}%
=(\sigma_{x_{j}p_{k}})_{1\leq j,k\leq n}.
\]
One checks, using the marginal properties (\ref{wigmarg}), that
\begin{align}
|\psi_{X,Y}(x)|^{2}  &  =\left(  \tfrac{1}{2\pi}\right)  ^{n/2}(\det
\Sigma_{XX})^{-1/2}\exp\left(  -\frac{1}{2}\Sigma_{XX}^{-1}x\cdot x\right)
\label{margauss1}\\
|\widehat{\psi}_{X,Y}(p)|^{2}  &  =\left(  \tfrac{1}{2\pi}\right)  ^{n/2}%
(\det\Sigma_{PP})^{-1/2}\exp\left(  -\frac{1}{2}\Sigma_{PP}^{-1}p\cdot
p\right)  . \label{margauss2}%
\end{align}
These conditions uniquely determine the matrices $\ \Sigma_{XX}$ and
$\Sigma_{PP}$. To find $\Sigma_{XP}$ one uses the Robertson--Schr\"{o}dinger
uncertainty principle in block-matrix form%
\begin{equation}
\Sigma_{XX}\Sigma_{PP}-(\Sigma_{XP})^{2}=\frac{\hbar^{2}}{4}I_{n\times n}
\label{RS}%
\end{equation}
which is equivalent to the symplectic conditions (\ref{cond12}) since
$G=\frac{\hbar}{2}\Sigma^{-1}$ is symplectic and symmetric. Equation
(\ref{RS}) however has multiple solutions $\Sigma_{XP}$. For instance, in the
case $n=1$(\ref{RS}) reduces to
\begin{equation}
\sigma_{xx}\sigma_{pp}-\sigma_{xp}^{2}=\frac{\hbar^{2}}{4} \label{n=1}%
\end{equation}
which has two solutions
\begin{equation}
\sigma_{xp}=\pm(\sigma_{xx}\sigma_{pp}-\frac{\hbar^{2}}{4})^{1/2}, \label{pp}%
\end{equation}
leading to two Gaussians $\psi^{\pm}$ forming what Corbett \cite{Corbett}
calls a \textquotedblleft Pauli partners\textquotedblright%
\begin{equation}
\psi^{\pm}(x)=\left(  \tfrac{1}{2\pi\sigma_{xx}}\right)  ^{1/4}e^{-\frac
{x^{2}}{4\sigma_{xx}}}e^{\frac{\pm i\sigma_{xp}}{2\hbar\sigma_{xx}}x^{2}}.
\label{Gauss1}%
\end{equation}
Thus, knowledge of the marginals of the Wigner transform $W\psi_{X,Y}$ does
nor allow us to recover the function $\psi_{X,Y}$ because we have a
\emph{phase intermediacy}.

\begin{remark}
The situation is similar for the general Radon transform with respect to an
arbitrary Lagrangian plane In fact, let $(\ell,\ell^{\prime})\in
\mathcal{F}_{\operatorname*{Lag}}(n)$. The Wigner transform $W\psi_{X,Y}$ is
determined by the knowledge of the two Radon transforms
\begin{align*}
R_{\ell}(W\psi_{X,Y})(z) &  =\int_{\ell(z)}W\psi_{X,Y}(\lambda)d\lambda\text{
}\\
R_{\ell^{\prime}}(W\psi_{X,Y})(z) &  =\int_{\ell^{\prime}(z)}W\psi
_{X,Y}(\lambda)d\lambda.
\end{align*}
The group $\operatorname*{Sp}(n)$ acting transitively on $\mathcal{F}%
_{\operatorname*{Lag}}(n)$ we can find $S\in\operatorname*{Sp}(n)$ such that
$(\ell,\ell^{\prime})=S(\ell_{X},\ell_{P})$. It is at this point the
transversality condition $\ell\cap\ell^{\prime}=0$ is esasential .
\end{remark}

Setting $S=%
\begin{pmatrix}
A & B\\
C & D
\end{pmatrix}
$ let $\widehat{S}\in\operatorname*{Mp}(n)$ be one of the two metaplectic
operators covering $S$ and
\[
\ell(Sz)=S\ell_{X}+Sz\text{ \ },\text{ \ }\ell^{\prime}(Sz)=S\ell_{P}+Sz.
\]
Using the symplectic covariance formula (\ref{wigcov}) we have
\begin{align*}
\int_{\ell^{\prime}(Sz)}W\psi_{X,Y}(\lambda)d\lambda &  =\int_{\mathbb{R}^{n}%
}W\psi_{X,Y}\left[  S((0,u)+(x,p))\right]  du\\
&  =\int_{\mathbb{R}^{n}}W(\widehat{S}^{-1}\psi_{X,Y})(x,u+p)du
\end{align*}
that is, in view of the first marginal formula (\ref{wigmarg}),
\begin{equation}
\int_{\ell^{\prime}(Sz)}W\psi_{X,Y}(\lambda)d\lambda=|\widehat{S}^{-1}%
\psi_{X,Y}(x)|^{2}.\label{spsi1}%
\end{equation}
Similarly, using the second formula (\ref{wigmarg}),
\begin{equation}
\int_{\ell(Sz)}W\psi_{X,Y}(\lambda)d\lambda=|F\widehat{S}^{-1}\psi
_{X,Y}(p)|^{2}.\label{spsi2}%
\end{equation}
These values allow the determination of $\widehat{S}^{-1}\psi_{X,Y}$ and,
hence, of $\psi_{X,Y}$ (up to an unimodular constant factor).

\subsection{Exact reconstruction using Bloch--Messiah diagonalization}

As the discussion above (formulas (\ref{margauss1})--(\ref{margauss2})) shows
two Lagrangian tomographies using the $x$ and $p$ planes do not suffice to
unambiguously reconstructing a Gaussian, because it leads to "Pauli partners".
The Radon inversion formula (\ref{inverse}) in Proposition \ref{Propinverse}
together with formula (\ref{integralbis}) dhow that, in principle, we need all
$|\widehat{U}\psi|$ for $U\in U(n)$ to reconstruct a function $\psi$.
Equivalently, we need all Lagrangian planes $\ell\in\operatorname*{Lag}(n)$ to
perform a complete tomography. It turns out that in the case of Gaussians we
need much less. The key to our argument the is following diagonalization
result which is a particular case of the \textit{Bloch--Messiah} \cite{Houde}
factorization result)

\begin{lemma}
[Bloch--Messiah]\label{LemmaBM}Let $G$ be a positive definite symmetric
matrix.. There exists $U\in U(n)$ such that$\hbar$
\begin{equation}
G=UG_{\Lambda}U^{T}\text{ \ \ , \ }^{T}G_{\Lambda}=%
\begin{pmatrix}
\Lambda & 0\\
0 & \Lambda^{-1}%
\end{pmatrix}
\label{bm}%
\end{equation}
where $\Lambda=\mathrm{diag}(\lambda_{1},\ldots,\lambda_{n})$ where
$\lambda_{1},\ldots,\lambda_{n}\leq1$ are the smallest eigenvalues of $G$.
\end{lemma}

\noindent See \cite{Birk,Houde} for proofs. (Recall \cite{Birk} that the
eigenvalues of a symmetric positive symplectic ,matrix occur in pairs
$(\lambda,\lambda^{-1}),$ $\lambda>0$, so if $\{\lambda_{1},\ldots,\lambda
_{n}\}$ \ is the set of smallest eigenvalues of $G$ then $\{\lambda_{1}%
,\ldots,\lambda_{n}$, $\lambda_{1}^{-1},\ldots,\lambda_{n}^{-1}\}$ is the set
of all eigenvalues of that matrix) .

Let $\ell\in\operatorname*{Lag}(n):Ax+Bp=0$; we have, as before, $\ell
=U^{T}\ell_{P}$ where $U=%
\begin{pmatrix}
A & B\\
-B & A
\end{pmatrix}
$. We have
\begin{equation}
R_{\ell}W\psi_{X,Y}(x,p)=|\widehat{U}\psi_{XY}(Ax+Bp)|^{2}%
\end{equation}
where $\widehat{U}\in\operatorname*{Mp}(n)$ covers the symplectic rotation $U$
. Writing
\[
W\psi_{XY}(z)=(\pi\hbar)^{-n}e^{-\frac{1}{\hbar}Gz\cdot z}\text{ \ \ },\text{
\ }G=%
\begin{pmatrix}
X+YX^{-1}Y & YX^{-1}\\
X^{-1}Y & X^{-1}%
\end{pmatrix}
\]
we have, using the Bloch--Messiah formula (\ref{bm}),
\begin{align*}
W(\widehat{U}\psi_{XY})(z)  &  =W\psi_{XY}(U^{T}z)\\
&  =(\pi\hbar)^{-n}e^{-\frac{1}{\hbar}UGU^{T}z\cdot z}\\
&  =(\pi\hbar)^{-n}e^{-\frac{1}{\hbar}(\Lambda x\cdot x+\Lambda^{-1}p\cdot
p)}.
\end{align*}
It follows that the LRT $R_{\ell}W\psi_{X,Y}$ is given by%
\[
R_{\ell}(W(\widehat{U}\psi_{X,Y})(z)=\left.  \left(  \tfrac{1}{\pi\hbar
}\right)  ^{n/2}(\det\Lambda)^{1/2}e^{-\tfrac{1}{\hbar}\Lambda y\cdot
y}\right\vert _{y=Ax+Bp}%
\]
which fully determines the diagonal matrix $\Lambda$ and hence $G=UG_{\Lambda
}U^{T}$ \ so we can explicit calculate
\begin{align*}
G  &  =%
\begin{pmatrix}
A & B\\
-B & A
\end{pmatrix}%
\begin{pmatrix}
\Lambda & 0\\
0 & \Lambda^{-1}%
\end{pmatrix}%
\begin{pmatrix}
A^{T} & -B^{T}\\
B^{T} & A^{T}%
\end{pmatrix}
\\
&  =%
\begin{pmatrix}
A\Lambda A^{T}+B\Lambda^{-1}B^{T} & A\Lambda B^{T}+B\Lambda^{-1}A^{T}\\
B\Lambda A^{T}+A\Lambda^{-1}B^{T} & B\Lambda B^{T}+A\Lambda^{-1}A^{T}%
\end{pmatrix}
.
\end{align*}
Solving the matrix equation \ \
\[%
\begin{pmatrix}
X+YX^{-1}Y & YX^{-1}\\
X^{-1}Y & X^{-1}%
\end{pmatrix}
=%
\begin{pmatrix}
A\Lambda A^{T}+B\Lambda^{-1}B^{T} & A\Lambda B^{T}+B\Lambda^{-1}A^{T}\\
B\Lambda A^{T}+A\Lambda^{-1}B^{T} & B\Lambda B^{T}+A\Lambda^{-1}A^{T}%
\end{pmatrix}
\]
in $X$ and $Y$ yields%
\[
X=\left(  B\Lambda B^{T}+A\Lambda^{-1}A^{T}\right)  ^{-1}\text{ ,
}Y=X(B\Lambda A^{T}+A\Lambda^{-1}B^{T})
\]
which determines the Gaussian $\psi_{X,Y}.$

We remark that the method outlined above is of a purely geometric nature since
it amounts to identifying the ellipsoid $Gz\cdot z\leq$ by the knowledge of
its semi-axes. It is more of a theoretical interest than a practical result
since it requires the knowledge of the symplectic rotation $U$ diagonalizing
$G$. Notice, however, that the eigenvalues of the covariance matrix $\Sigma$
and hence those of $G=\frac{\hbar}{2}\Sigma^{-1}$ do only depend on $G$ and
can be determined as soon as the marginals
\begin{align*}
R_{\ell_{X}}(W(\widehat{U}\psi_{X,Y})  &  =\left(  \tfrac{1}{2\pi}\right)
^{n/2}(\det\Sigma_{XX})^{-1/2}\exp\left(  -\frac{1}{2}\Sigma_{XX}^{-1}x\cdot
x\right) \\
R_{\ell_{P}}(W(\widehat{U}\psi_{X,Y})  &  =\left(  \tfrac{1}{2\pi}\right)
^{n/2}(\det\Sigma_{PP})^{-1/2}\exp\left(  -\frac{1}{2}\Sigma_{PP}^{-1}p\cdot
p\right)
\end{align*}
and are known (formulas (\ref{margauss1}) and (\ref{margauss2})).This is
because the characteristic equation of $\Sigma$ depends on the product
$\Sigma_{XP}^{2}$ only.

\section{Application to Density Operators}

Hilbert--Schmidt operators on $L^{2}(\mathbb{R}^{n})$ are exactly those with
square integrable kernel. In view of the fundamental relation
\begin{equation}
a(x,p)=\int_{\mathbb{R}^{n}}e^{-\frac{i}{\hbar}p\cdot y}K(x+\tfrac{1}%
{2}y,x-\tfrac{1}{2}y)dy\text{.} \label{AK9}%
\end{equation}
between the kernel of an operator $A$ and its Weyl symbol, an operator
$\mathbf{A=}\operatorname*{Op}_{\mathrm{Weyl}}\mathbf{(a)}$ is Hilbert Schmidt
if and only if $a\in L^{2}(\mathbb{R}^{2n})$ as follows from the identity%
\begin{equation}
||a||_{L^{2}(\mathbb{R}^{2n})}=\left(  2\pi\hbar\right)  ^{n/2}||K||_{L^{2}%
(\mathbb{R}^{n}\times\mathbb{R}^{n})}\text{.} \label{kanorm10}%
\end{equation}

\subsection{Wigner bases}

Recall that the cross-Wigner transform of $\psi,\phi\in L^{2}(\mathbb{R}^{n})$
is defined by
\begin{equation}
W(\psi,\phi)(z)=\int e^{-\frac{i}{\hbar}p\cdot y}\psi(x+\tfrac{1}%
{2}y)\overline{\phi(x-\tfrac{1}{2}y)}dy. \label{crosswig}%
\end{equation}
For $\psi=\phi$ it reduces to the usual Wigner transform (\ref{Wigner}). The
mapping $(\psi,\phi)\rightarrow W(\psi,\phi)$ is sesquilinear ad ha%
%TCIMACRO{\U{b4}}%
%BeginExpansion
\'{}%
%EndExpansion
the following main properties: $W(\psi,\phi)=\overline{W(\phi,\psi)}$ and%
\begin{equation}
\int W(\psi,\phi)(z_{0})dz_{0}=(\psi|\phi)_{L^{2}}. \label{wl1l2}%
\end{equation}
The latter follows from the marginal properties%
\begin{align}
\int W(\psi,\phi)(z)dp  &  =\psi(x)\overline{\phi(x)}\text{\ }%
\label{wigjoint1}\\
\int W(\psi,\phi)(z)dx  &  =\widehat{\psi}(p)\overline{\widehat{\phi}(p)}
\label{wigjoint2}%
\end{align}

The following result seems to be implicitly known in the literature. We have
given a detailed proof of it in \cite{Birkbis}:

\begin{proposition}
\label{ONB}Let $(\psi_{j})_{j}$ and $(\phi_{k})_{k}$ \ be arbitrary
orthonormal bases of $L^{2}(\mathbb{R}^{n})$; the vectors
\[
W_{j,k}=(2\pi\hbar)^{n/2}W(\psi_{j},\phi_{k})
\]
form an orthonormal basis of $L^{2}(\mathbb{R}^{2n})$. In particular
$(W_{\psi_{j}}\psi_{k})_{j,k}$ is such a basis.
\end{proposition}

\begin{proof}
Using Moyal's identity \cite{Birkbis,WIGNER}we have
\begin{align*}
(W_{j,k}|W_{j^{\prime},k^{\prime}})_{L^{2}}  &  =(2\pi\hbar)^{n}(W(\psi
_{j},\phi_{k})|W(\psi_{j^{\prime}},\phi_{k^{\prime}}))_{L^{2}}\\
&  =(\psi_{j}|\psi_{j^{\prime}})_{L^{2}}\overline{(\phi_{k}|\phi
_{k^{\prime\prime}})_{L^{2}}}%
\end{align*}
hence the $W_{j,k}$ form an orthonormal system in $L^{2}(\mathbb{R}^{2n})$.
Let us show that the $W_{j,k}$ form a complete system (for a different proof
see our monograph \cite{Birkbis}, \S 19.1). Let $a\in L^{2}(\mathbb{R}^{2n})$.
\ and consider the Hilbert--Schmidt operator $A=\operatorname*{Op}%
_{\mathrm{Weyl}}(a)$. We have
\[
(A\psi_{j}|\phi_{k})_{L^{2}}=\int_{\mathbb{R}^{2n}}a(z)W(\psi_{j},\phi
_{k})(z)dz
\]
hence
\[
(a|W(\psi_{j},\phi_{k}))_{L^{2}}=0
\]
for all indices $j,k$ implies that $A\psi_{j}=0$ for all $j$ hence $\ A=0$ and
$a=0$.
\end{proof}

A simple illustration is provided by the Hermite functions and their
cross-Wigner transforms, the Laguerre functions. Recall that the $j$-th
Hermite function is the are the $j$-th eigenfunctions of the operator
$\frac{1}{2}\left(  -\hbar^{2}\partial_{x}^{2}+x^{2}\right)  $it is given by%
\[
\psi_{j}(x)=\hbar^{-1/4}H_{j}(\hbar^{-1/2}x)
\]
where $H_{j}$ and $h_{j}$ are the functions%
\begin{align*}
H_{j}(x)  &  =\sqrt{\frac{1}{2^{j}j!}}\left(  \frac{1}{\pi}\right)
^{1/4}e^{-x^{2}/2}h_{j}\left(  x\right) \\
h_{N}(x)  &  =(-1)^{j}e^{x^{2}}\left(  \frac{d^{j}}{dx^{j}}e^{-x^{2}}\right)
.
\end{align*}
The Hermite functions form an orthonormal basis of $^{2}(\mathbb{R})$. The
cross-Wigner tenesmuses of the $\psi$ are the Laguerre functions
\cite{WIGNER}, given, for $j<k$, by%
\[
W(\psi_{j},\psi_{k})(z)=C_{j,k}\left(  x\text{+}ip\right)  ^{k-j}e^{-\frac
{2}{\hbar}H(z)}L_{j}^{(k-j)}\left(  \frac{4}{\hbar}H(z)\right)
\]
where the constant factor $C_{j,k}$ is given by$C_{j,k}C_{j,k}$%
\[
C_{j,k}=\frac{(-1)^{j}}{\pi\hbar}\sqrt{\frac{2^{k}j!}{2^{j}k!}}\left(
\frac{1}{\hbar}\right)  ^{(k-j)/2}.
\]
The case $j>k$ is obtained using the conjugation formula $W(\psi
,\phi)=\overline{W(\phi,\psi)}$ and%
\[
W(\psi_{j},\psi_{j})(z)=C_{j,j}e^{-\frac{2}{\hbar}H(z)}L_{j}^{0}\left(
\frac{4}{\hbar}H(z)\right)
\]
with $H=\frac{1}{2}(p^{2}+x^{2})$..

\subsection{Application: Hilbert--Schmidt and density operators}

A density operator on $L^{2}(\mathbb{R}^{n})$ is a self-adjoint positive
definite operator with trace one. Density operators are de facto
Hilbert--Schmidt ope3raors; each density operators is in fact the product of
two Hilbert--Schmidt pourers. Let $A=\operatorname*{Op}_{\mathrm{Weyl}}(a)$ be
of trace class; thus $a\in L^{2}(\mathbb{R}^{2n})$. If in addition $a\in
L^{\substack{2\\1}}(\mathbb{R}^{2n})$ then
\begin{equation}
\operatorname*{Tr}(A)=\left(  \frac{1}{2\pi\hbar}\right)  ^{n}\int a(z)dz.
\label{tra1}%
\end{equation}
A \textit{caveat}: the relation $\int a(z)dz.<\infty$ alone does not imply
that $A=\operatorname*{Op}_{\mathrm{Weyl}}(a)$ is of trace class: see our
discussion in \cite{Birkbis}, \S 12.3. Applying the spectral theorem for
self-adjoint compact operators we see that if $A$ is a density operator there
exists an orthonormal system of vectors $(\psi_{j})_{j}$ in $L^{2}%
(\mathbb{R}^{n})$ and nonnegative scalars $(\lambda_{j})_{j}$ such that
\begin{equation}
A=\sum_{j}\lambda_{j}\Pi_{\psi_{j}}\text{ \ },\text{ }\operatorname*{Tr}%
A=\sum_{j}\lambda_{j}=1 \label{spec}%
\end{equation}
where $\Pi_{\psi_{j}}$ is the orthogonal projection of $L^{2}(\mathbb{R}^{n})$
onto the subspace generated by $\psi_{j}$ (counting multiplicities). The
$\lambda_{j}$ are the eigenvalues of $A$. It follows that the Weyl symbol of
$A$ is given by
\begin{equation}
a=\left(  \frac{1}{2\pi\hbar}\right)  ^{n}\sum_{j}\lambda_{j}W\psi_{j}.
\label{symba}%
\end{equation}
It is customary in quantum mechanics to call
\begin{equation}
\rho=(2\pi\hbar)^{n}a=\sum_{j}\lambda_{j}W\psi_{j} \label{wigdi}%
\end{equation}
the "Wigner distribution" of the density operator $A$, in which case trace
formula (\ref{tra1}) reads $\operatorname*{Tr}\rho=1$. It should be pointed
out that the splitting of the Wigner distribution in elementary states is not
unique. In fact, Let $\rho$ be the Wigner distribution of a density operator
$A=\operatorname*{Op}_{\mathrm{Weyl}}(a)$ and $(\psi_{j})_{j}$ an orthonormal
basts of $L^{2}(\mathbb{R}^{n})$. There exists a family of constants $c_{j}$
such tat \
\[
\rho=\sum_{j}c_{j}W\psi_{j}+\sum_{j\neq k}c_{jk}W(\psi_{j},\psi_{k})
\]
we have, using (\ref{wl1l2})
\begin{align*}
\operatorname*{Tr}A  &  =\int\rho(z)dz=\sum_{j}c_{j}\\
c_{j}  &  =(2\pi\hbar)^{n}(\rho|W\psi_{j})_{L^{2}}.
\end{align*}

In practice, density operators are obtained as mixtures states $\phi_{j}$ each
having probability $\rho_{j}$ to occur, and one then defines the Wigner
distribution of this mixed state by $\rho=^{n}\sum_{j}\rho_{j}\phi_{j}$ \ and
one easily shows \cite{Birkbis,WIGNER} that%
\[
A=(2\pi\hbar)^{n}\operatorname*{Op}\nolimits_{\mathrm{Weyl}}(\rho)
\]
is a density operator.

We extend definition (\ref{integral}) of the Lagrangian Radon transform to the
cross Wigner distribution in the obvious way: by definition%
\[
R_{\ell}W/(\psi,\phi)(z)=\int_{\mathbb{R}^{n}}W/(\psi,\phi)((-B^{T}%
u+x,A^{T}u+p)du
\]
and we will write%
\[
R_{\ell}W/(\psi,\phi)(z)=\int_{\ell(z)}W/(\psi,\phi)(\lambda)d\lambda.
\]
Using an argument similar to that of the proof of Proposition \ref{Proprawig}
above one shows \ that for $\psi,\phi\in L^{1}(\mathbb{R}^{n})\cap
L^{2}(\mathbb{R}^{n})$, we have
\begin{equation}
R_{\ell}W(\psi,\phi)(x,p)=\widehat{U}\psi(Ax+Bp)\overline{\widehat{U}%
\phi(Ax+Bp)} \label{pfu}%
\end{equation}
where $\ell\in\operatorname*{Lag}(n):Ax+Bp$. where $\widehat{U}\in
\operatorname*{Mp}(n)$ covers the symplectic rotation $U$ taking $\ell$ to
$\ell_{P}$. The Lagrangian Radon transform of the Wigner distribution of a
density matrix is then defined by superposition. For instance, the spectral
decomposition (\ref{wigdi}) yields
\begin{equation}
R_{\ell}\rho=\sum_{j}\lambda_{j}R_{\ell}W\psi_{j}.
\end{equation}
More generally:

\begin{proposition}
Let $A=\operatorname*{Op}_{\mathrm{Weyl}}(a)$ be a density operator, $a\in
L^{1}(\mathbb{R}^{2n})\cap L^{2}(\mathbb{R}^{2n}).$Let $(\psi_{j})_{j}$ and
$(\phi_{k})_{k}$ \ be two orthonormal bases of $L^{2}(\mathbb{R}^{n})$; The
Lagrangian Radon transform of its Wigner distribution is
\begin{equation}
R_{\ell}\rho=\sum_{j,k}c_{jk}\widehat{U}\psi_{j}(Ax+Bp)\overline
{\widehat{U}\phi_{k}(Ax+Bp)} \label{roll}%
\end{equation}
where $(c)$ is a sequence such that
\begin{equation}
c_{jk}=(2\pi\hbar)^{n}\int\rho(z)W(\psi_{j},\phi_{k})(z)dz. \label{cjk}%
\end{equation}

\end{proposition}

\begin{proof}
We have
\[
(A\psi_{j}|\phi_{k})_{L^{2}}=\int_{\mathbb{R}^{2n}}a(z)W(\psi_{j},\phi
_{k})(z)dz=c_{jk};
\]
since $A$ is of trace class , $\sum_{j,k}|c_{jk}|<\infty$ (\cite{Birkbis},
Prop. 277) .We have%
\[
\rho=\sum_{j,k}c_{jk}W(\psi_{j},\phi_{k})
\]
hence (\ref{roll}) in view of (\ref{pfu}).
\end{proof}

One may therefore reconstruct $\rho$ in two steps: (1) Fourier decomposition
of the tomogram ; (2) expansion of each Fourier component in the considered
bases. \ From this viewpoint, the family of cross-Wigner distributions
$\{W(f_{j},g_{k})\}$ plays the role of a \textquotedblleft quantum Fourier
basis\textquotedblright\ for phase-space tomography. Rather than
reconstructing the Wigner function pointwise by the inverse Radon transform,
one reconstructs directly the coefficients $c_{jk}$ and hence the density
operator itself.

\begin{acknowledgement}
This work has been financed by the Austrian Research Foundation FWF (Grant
number PAT 2056623).
\end{acknowledgement}

\textbf{DATA\ AVAILABILITY\ STATEMENT}: no data has been used or created.

\textbf{CONFLICT\ OF\ INTERESTS}: there are no conflicts of interests

maurice.de.gosson@univie.ac.at
\end{document}